\documentclass[pdftex,pra,aps,onecolumn,twoside,nofootinbib,showkeys, showpacs]{revtex4}
\pdfoutput=1
\usepackage[utf8]{inputenc}
\usepackage[dvips]{graphicx}
\usepackage[stable]{footmisc}
\usepackage{amsmath,amsthm,amssymb}
\usepackage{youngtab}
\usepackage{mathtools}
\usepackage[all]{xy}
\usepackage{dsfont}
\usepackage{multirow}
\usepackage{rotating}
\usepackage{appendix}
\usepackage{color}
\def\squareforqed{\hbox{\rlap{$\sqcap$}$\sqcup$}}
\def\qed{\ifmmode\squareforqed\else{\unskip\nobreak\hfil
\penalty50\hskip1em\null\nobreak\hfil\squareforqed
\parfillskip=0pt\finalhyphendemerits=0\endgraf}\fi}

\def\duzomniejsze{<\kern-.7mm<}
\def\duzowieksze{>\kern-.7mm>}

\def\textbf#1{{\bf #1}}
\def\beq{\begin{equation}}
\def\eeq{\end{equation}}
\def\be{\begin{equation}}
\def\ee{\end{equation}}
\def\bal{\begin{align}}
\def\eal{\end{align}}
\def\ben{\begin{eqnarray}}
\def\een{\end{eqnarray}}
\def\beqa{\begin{eqnarray}}
\def\eeqa{\end{eqnarray}}
\def\eea{\end{array}}
\def\bea{\begin{array}}

\newcommand{\bei}{\begin{itemize}}
\newcommand{\eei}{\end{itemize}}
\newcommand{\bee}{\begin{enumerate}}
\newcommand{\eee}{\end{enumerate}}

\def\tr{{\rm Tr}}

\def\>{\rangle}
\def\<{\langle}

\def\ot{\otimes}

\newtheorem{lemma}{Lemma}
\newtheorem{corollary}{Corollary}
\newtheorem{proposition}{Proposition}
\newtheorem{theorem}{Theorem}
\newtheorem{definition}{Definition}

\newtheorem{remark}{Remark}

\theoremstyle{plain}

\def\bed{\begin{definition}}
\def\eed{\end{definition}}
\def\bel{\begin{lemma}}
\def\eel{\end{lemma}}
\def\bet{\begin{theorem}}
\def\eet{\end{theorem}}

\newcommand{\im}{\operatorname{Im}}

\newcommand{\R}{R^{-1}}

\begin{document}

\title{Using non positive maps to characterize entanglement witnesses}
\author{Marek Mozrzymas$^1$, Adam Rutkowski$^{2,3}$, Micha{\l } Studzi{\'n}%
ski$^{2,3}$}
\affiliation{$^1$ Institute for Theoretical Physics, University of Wroc{\l }aw, 50-204
Wroc{\l }aw, Poland\\
$^2$ Faculty of Mathematics, Physics and Informatics, University of Gda{\'n}%
sk, 80-952 Gda{\'n}sk, Poland\\
$^3$ Quantum Information Centre of Gda\'{n}sk, 81-824 Sopot, Poland}
\date{\today }

\begin{abstract}
In this paper we present a new method for entanglement witnesses construction. We show that to construct such an object we can deal with maps which are not positive on the whole domain, but only on a certain sub-domain. In our approach crucial role play such maps which are surjective between sets $\mathcal{P}_{k}^d$  of $k \leq d$ rank projectors and the set $\mathcal{P}_1^d$ of rank one projectors acting in the $d$ dimensional space. We argue that our method can be used to check whether a  given observable is an entanglement witness. In the second part of this paper we show that inverse reduction map satisfies this requirement and using it we can obtain a bunch of new entanglement witnesses.
\end{abstract}

\maketitle



\section{Introduction}
It is well known that quantum entanglement is the most important  resource in the field of quantum information theory. It is worth to mention here such significant achievements as quantum cryptography~\cite{Bennett}, quantum teleportation~\cite{Bennett2}, quantum dense coding~\cite{Bennett3}, quantum error corrections codes and many other important applications of this phenomena.
That is why, the knowledge we gain from dealing with entangled states with their classification is of priority importance.
However, still one of the biggest problems  in the field remains open. Namely, up to now we do not have satisfactory criteria to decide whether a given quantum state is separable or entangled. A full answer is delivered by a famous Peres-Horodecki criterion~\cite{Hor1,Peres} based on the idea of partial transposition, which gives necessary and sufficient criteria for separability for bipartite $2 \ot 2, 2 \ot 3$ systems, but unfortunately for higher dimensions this criterion is not conclusive. The problem is even more complicated if we lift it to multipartite case, but of course there are several approaches to detect entanglement or checking separability in general~\cite{Doherty1, Doherty2,Guhne}.  Despite these difficulties, fortunately there is one   most general method to decide when quantum composite state is entangled. It is based on the concept of an entanglement witness firstly introduced in~\cite{Terhal} making use of the famous Hahn-Banach theorem. This approach allows us to detect entanglement  without full knowledge about the quantum state. What is the most important any entangled state has a corresponding entangled witness, so this property makes the mentioned method somehow universal. Exploring theory of entanglement witnesses from a mathematical point of view, there is a well known connection between them and the theory of positive maps~\cite{PH}, which allows us to understand much deeper the structure of the set of quantum states.

Let us say here a few words more about notation used in this manuscript. In this section and also in our further considerations  by $\mathcal{B}(\mathcal{\mathbb{C}}^d)$ (respectively $\mathcal{B}(\mathcal{H})$) we denote the algebra of all bounded linear operators on $\mathbb{C}^d$(respectively on $\mathcal{H}$). Using this notation let us define the following set:
\be
\label{Qset}
\mathcal{S}(\mathcal{H})=\{\rho \in \mathcal{B}(\mathcal{H}) \ | \ \rho \geq 0, \tr \rho=1\},
\ee
which is set of all states on space $\mathcal{H}$. Suppose now that we are dealing with two finite dimensional Hilbert spaces $\mathcal{H},\mathcal{K}$. State in the bipartite composition system $\rho \in \mathcal{S}(\mathcal{H}\ot \mathcal{K})$ is said to be separable if it can be written as $\rho=\sum_i p_i \rho_i \ot \sigma_i$, where $\rho_i,\sigma_i$ are states on $\mathcal{H}$ and $\mathcal{K}$ respectively, and $p_i$ are some positive numbers satisfying $\sum_i p_i=1$. Otherwise we say that state $\rho$ is entangled.

Now we are ready to present the definition of entanglement witness and basic ideas connected with these objects. Let us start from the definition of entanglement witness~\cite{Hor1},~\cite{Terhal}:

\begin{definition}
\label{def01}
The operator $W \in \mathcal{B}\left(\mathbb{C}^d\ot \mathbb{C}^d\right)$ is called entanglement witness when:
\begin{enumerate}
\item $W \ngeq  0$,
\item $\tr \left(\sigma W\right)\geq 0$ for all separable states $\sigma$.
\end{enumerate}
\end{definition}

There is a well known theorem~\cite{Hor1} which states that for every entangled state $\rho$ there exists a corresponding entangled witness $W$, such that $\tr(W\rho)<0$. Reader notices that this condition is equivalent to the first condition from the above definition.  From Definition~\ref{def01} we see that any entanglement witness corresponds to a hermitian operator, which thanks to Jamio{\l}kowski isomorphism~\cite{Jam} is connected with some positive, but not completely positive linear map $\Lambda: \mathcal{B}\left(\mathbb{C}^d\right) \rightarrow \mathcal{B}\left(\mathbb{C}^d\right)$, such that:
\be
\label{jamW}
W=\left(\mathbf{1}\ot \Lambda \right) P_d^+,
\ee
where $P_d^+$ is the projector on maximally entangled state $|\psi_+^d\>=(1/\sqrt{d})\sum_i |ii\>$.

 At this point for more information about entanglement witnesses and their properties we refer the reader to an excellent review paper treating this topic~\cite{chrust}. Unfortunately, definitions of entanglement witnesses are not really efficient in practice. Namely, to check whether a given observable $W$ is an entanglement witness we have to find a positive, but not completely positive linear map $\Lambda$, which after acting on the half of maximally entangled state gives as a result the operator $W$. Clearly,  finding such maps is a hard task and only in a few cases we can find their  form. The second method is checking the block positivity of the operator $W$, which is an extremely hard and time consuming task, because there is no general method of dealing with this problem. In this manuscript  we present a new approach to check when a given observable $W$ is an entanglement witness. This method is based on the idea of non positive maps. Of course our approach is not fully general, since we do not have a full characterization of non positive maps in the sense which we explain further on, but in our opinion it opens new opportunities in the field.


At the end of this introductory section we present the structure of our paper.
 Namely in Section~\ref{results} the main result of our work is contained.  In the Theorem~\ref{prop6} we show that to construct an entanglement witness we do not have to restrict to positive maps on the whole domain, but only on its certain subset. In particular such a map has to be at least a surjection between set of the rank $k \leq d$ projectors $\mathcal{P}_k^d$ and set of rank one projectors $\mathcal{P}_1^d$ acting in the $d$ dimensional space.

After that we present two short sections with examples which illustrate how our method works in practice. We start from the Section~\ref{reduction} where we show that the inverse reduction map satisfies all requirements from the Section~\ref{results}, then in the Section~\ref{examples} we show an  illustrative example of entanglement witnesses obtained thanks to the inverse reduction map.

Finally, we present Appendix~\ref{AppA} where we explain the basic properties of unitary spaces which are necessary to discuss inverse reduction map in the Section~\ref{reduction}. In the Appendix~\ref{AppB} we formulate Propositions~\ref{prop3} and~\ref{prop5} which together with the Remark~\ref{RemP} are necessary in the proof of Theorem~\ref{prop6} and the formulation itself play a very important role in the analysis of the inverse reduction map from the Section~\ref{reduction}.

%


\section{General construction of entanglement witness from non-positive map}
\label{results}
In this section we present our main result contained in the Theorem~\ref{prop6}. We show that to construct entanglement witnesses we do not have to restrict only to positive maps on the whole domain in general, but only on some specific subset. To do so, we can use map $\Lambda^{\dagger}: \mathcal{B}\left(\mathbb{C}^d\right) \rightarrow \mathcal{B}\left(\mathbb{C}^d\right)$, which is surjective  between set $\mathcal{P}_{k}^{d}$ of rank $k$ projectors and the set $\mathcal{P}_{1}^{d}$ of rank one projectors, which is given in the Proposition~\ref{prop3} contained in the Appendix~\ref{AppA}. Having this knowledge we are in the position to formulate the following:


\begin{theorem}
\label{prop6}
Let $W\in \mathcal{B}\left(\mathbb{C}^d\ot \mathbb{C}^d\right)$, $W=W^{\dagger}$ , $W\ngeq 0$ and $W$ is such that $\widetilde{W}=(\mathbf{1}%
\otimes \Lambda)W\geq 0$ for some linear map $\Lambda: \mathcal{B}\left(\mathbb{C}^d\right) \rightarrow \mathcal{B}\left(\mathbb{C}^d\right)$. We assume that the map $\Lambda^{\dagger}: \mathcal{B}\left(\mathbb{C}^d\right) \rightarrow \mathcal{B}\left(\mathbb{C}^d\right)$ is not positive on the whole domain but only maps surjectivley set $\mathcal{P}_{k}^{d}$ of rank $k$ projectors on the set $\mathcal{P}_{1}^{d}$ of rank one projectors, then we have:
\be
\forall |\psi\>,|\phi\> \in\mathbb{C}^{d}:||\psi||=||\phi||=1 \quad 0\leq \tr(W |\psi\>\<\psi|\otimes |\phi\>\<\phi|),
\ee
so the operator $W$ is an entanglement witness.
\end{theorem}

\begin{proof}
Let $W\in \mathcal{B}\left(\mathbb{C}^d\right)$, $W=W^{\dagger}$ , $W\ngeq 0$ and $W$ is such that $\widetilde{W}=(\mathbf{1}%
\otimes \Lambda)W\geq 0$, so we can write~\cite{Horn}:
\be
0\leq \sum_{i=1}^{k}\lambda _{i}=\min_{\widetilde{P}\in \mathcal{P}%
_{k}^{d^{2}}}\tr(\widetilde{W}\widetilde{P}),
\ee
where $0\leq \lambda _{1}\leq \lambda _{2}\leq ....\leq \lambda _{d-1}\leq
...\leq \lambda _{d^{2}}$ are eigenvalues of $\widetilde{W}$ and here $%
\mathcal{P}_{k}^{d^{2}}=\{P\in \mathcal{B}\left(\mathbb{C}^d\ot \mathbb{C}^d\right):P^{2}=P,\quad P^{\dagger}=P,\quad \tr(P)=k\}$. Now, we choose a particular
orthogonal projector $P=|\psi\>\<\psi|\otimes
\sum_{i=1}^{k}|\phi_i\>\<\phi_i|=\sum_{i=1}^{k}|\omega_i\>\<\omega_i|$ from the
Proposition~\ref{prop5} we get
\be
\label{as}
0\leq \sum_{i=1}^{k}\lambda _{i}=\min_{\widetilde{P}\in \mathcal{P}
_{k}^{d^{2}}}\tr(\widetilde{W}\widetilde{P})\leq \tr(\widetilde{W}P)=\tr((
\mathbf{1}\otimes \Lambda)WP).
\ee
Now, we can continue rewriting the right hand side of the formula~\eqref{as} as
\be
\label{5}
0\leq \tr((\mathbf{1}\otimes \Lambda)WP)=\tr(W(\mathbf{1}\otimes \Lambda^{\dagger})|\psi\>\<\psi|\otimes
\sum_{i=1}^{k}|\phi_i\>\<\phi_i|)=\tr(W|\psi\>\<\psi|\otimes
\Lambda^{\dagger}(\sum_{i=1}^{k}|\phi_i\>\<\phi_i|))),
\ee
where by $\Lambda^{\dagger}$ we denote the adjoint~\footnote{Suppose that we are given the liner map $\Lambda: \mathcal{B}\left(\mathbb{C}^d\right)\rightarrow \mathcal{B}\left(\mathbb{C}^d\right)$, then the adjoint map is defined as $
\tr\left(A \Lambda(B)\right)=\tr\left(B \Lambda^{\dagger}(A)\right), \quad \forall \ A,B \in \mathcal{B}\left(\mathbb{C}^d\right).
$
We say that the linear map $\Lambda$ is self-adjoint when $A,B \in \mathcal{B}\left(\mathbb{C}^d\right) \
\tr\left(A \Lambda(B)\right)=\tr\left(B \Lambda(A)\right).
$ Moreover, if $\Lambda$ is a positive map then  $\Lambda^{\dagger}$ is a positive map, too.} The projector $P$ is chosen from all projectors $\mathcal{P}_k^{d^2}$ in such a way that we get the obvious second inequality in equation~\eqref{as}, but first of all we get a desired formula on RHS of eq.~\eqref{5}. In this proof the role of $P$ is technical and purely auxiliary. The projectors $\sum_{i=1}^{k}|\phi_i\>\<\phi_i|$ of rank $k$
generate the set $\mathcal{P}_{k}^{d}=\{P\in \mathcal{B}\left(\mathbb{C}^d\right):P^{2}=P,\quad P^{\dagger}=P\qquad \tr(P)=k\}$,  thus from the assumptions it
follows that $\{\Lambda^{\dagger}(\sum_{i=1}^{k}|\phi_i\>\<\phi_i|):\<\phi_{i}|\phi_{j}\>=\delta
_{ij}\}=\mathcal{P}_{1}^{d}=\{P\in \mathcal{B}\left(\mathbb{C}^d\right):P^{2}=P,\quad P^{\dagger}=P\qquad \tr(P)=1\}$. It means that $W\in \mathcal{B}\left(\mathbb{C}^d\ot \mathbb{C}^d\right)$, $W=W^{\dagger}$, $W\ngeq 0$ takes non-negative expectation values on
separable states. This finishes the proof.
\end{proof}

\begin{remark}
\label{marek1}
Form the proof of Theorem~\ref{prop6} it follows that the operator $W$ cannot take negative values on the product states.
\end{remark}

	\begin{remark}
		One can see that we can use statement from the Theorem~\ref{prop6} to chceck whether a given $W$ is an entanglement witness. Namely, for a given $W$ which fulfills assumptions it is enough to find such a map $\Lambda^{\dagger}: \mathcal{B}\left(\mathbb{C}^d\right) \rightarrow \mathcal{B}\left(\mathbb{C}^d\right)$ acting surjectivley between set $\mathcal{P}_k^d$ and $\mathcal{P}_1^d$ for which we have $(\mathbf{1}%
		\otimes \Lambda)W\geq 0$.
	\end{remark}
In the next paragraph we show that the inverse reduction map fulfils all required assumptions and we can use it to check if an observable $W$ is an entanglement witness without checking block-positivity, which is a hard task in general. Of course in the case when for a given $W$ condition $(\mathbf{1}%
\otimes \Lambda)W\geq 0$ is not satisfied our method is not conclusive.

\section{Inverse reduction map as an example}
\label{reduction}
In the previous section we have considered general maps $\Lambda$ with certain properties. Naturally a question arises: Do we know any examples of the maps which satisfy required demands? The goal of this paragraph is to present such an example. Let us consider the following linear map

\begin{definition}
\label{def1}
\be
\R:\mathcal{B}\left(\mathbb{C}^d\right)\rightarrow \mathcal{B}\left(\mathbb{C}^d\right),\quad \quad \forall A\in \mathcal{B}\left(\mathbb{C}^d\right) \quad \R(A)=\frac{1}{d-1}\tr(A)\mathbf{1}-A.
\ee
\end{definition}

The linear map $\R$ acts in the linear space $\mathcal{B}\left(\mathbb{C}^d\right)$ which is  a Hilbert space with respect to the standard Hilbert-Schmidt scalar
product
\be
\forall A,B\in \mathcal{B}\left(\mathbb{C}^d\right)\quad (A,B)\equiv \tr(A^{\dagger}B),
\ee
where $\dagger$ is the hermitian conjugation.

\begin{remark}
\label{rem01}
The reduction map~\cite{Hor3} is defined as
\be
R: \mathcal{B}(\mathcal{H})\rightarrow \mathcal{B}(\mathcal{H}), \quad \forall A\in \mathcal{B}(\mathcal{H}) \quad R(A)=\tr(A)\mathbf{1}-A.
\ee
Indeed Reader can check that for arbitrary $A\in \mathcal{B}(\mathcal{H})$ we have $\R \circ R(A)=R \circ \R(A)=A$, which means that  $\ker \left(\R \right)=\{0\},$ and we have $\forall A\in \mathcal{B}\left(\mathbb{C}^d\right)\quad R(A)=\tr(A)\mathbf{1}-A$. Note that if $d=2$, then $R=\R$. In this case both maps $R,\R$ are positive.
\end{remark}

\begin{proposition}
\label{prop4}
The map $\R$ has the following properties:
\begin{enumerate}
\item The map $\R$ is self-adjoint with respect to the Hilbert-Schmidt scalar product
i.e. we have
\be
(\R(A),B)=(A,\R(B)).
\ee


\item Suppose that $A\in \mathcal{B}(\mathbb{C}^d)$ is an orthogonal projector of rank $d-1$ i.e. $A^{2}=A,\quad A^{\dagger}=A$ and
$\tr(A)=d-1,$ then
\be
\R(A)=\mathbf{1}-A\Rightarrow \R(A)^{2}=\R(A),\quad \R(A)^{\dagger}=\R(A),\quad
\tr \R(A)=1,
\ee
so the image of the map $\R$ on any orthogonal projector of rank $d-1$ is an
orthogonal projector of rank $1$. Moreover, the map $\R$ establishes a
bijective correspondence between the set of all orthogonal projectors of rank
$d-1$ and the set of all orthogonal projectors of rank $1$.
\item If $\forall X\in \mathcal{B}\left(\mathbb{C}^d\right)$ we have $R^{-1}(X)\geq 0$ then $X\geq 0$.
\end{enumerate}
\end{proposition}
The proof of above statements can be directly deduced from the facts contained in the Appendix~\ref{AppA}.


\begin{corollary}
\label{cor2}
The operator $\mathbf{1}\otimes \R$ is also self-adjoint with respect to the tensor product scalar product
\be
\forall A,B,X,Y\in \mathcal{B}\left(\mathbb{C}^d\right) \quad (A\otimes B,X\otimes Y)\equiv (A,X)(B,Y),
\ee
where $(A,B)\equiv \tr(A^{\dagger}B)$.
\end{corollary}

\begin{remark}
\label{rem1}
The map $\R:\mathcal{B}\left(\mathbb{C}^d\right)\rightarrow \mathcal{B}\left(\mathbb{C}^d\right)$ is not positive but $\R$ restricted to the set $\mathcal{P}_{d-1}^d$ is a positive map. Indeed, as an example let us take matrix $\mathbb{I}$ which is filled only by ones and it is positive.
Now acting by $\R$ we have $A=\R(\mathbb{I})=\frac{d}{d-1}\mathbf{1}-\mathbb{I}$, with $\operatorname{spec}(A)=\left\{\frac{d(2-d)}{d-1},\frac{d}{d-1},\ldots,\frac{d}{d-1}\right\}$.
We notice that whenever $d>2$, then $\frac{d(2-d)}{d-1}<0$, so $A$ is no longer positive.
Summarizing, when $R,\R$ are not equal (for $d \geq 3$, see Remark~\ref{rem01}), then the main difference between them is that $R$ is positive but $\R$ it is not in general.
\end{remark}
 To sum up, the  inverse reduction map $\R$  satisfies all conditions from the assumptions of the Theorem~\ref{prop6}, so it can be used for the entanglement witness construction. Moreover, thanks to Proposition~\ref{prop4} point 1) this map is self-adjoint, so it satisfies even stronger conditions than we require.

\section{Explicit examples of entanglement witnesses}
\label{examples}
In this section we use Theorem~\ref{prop6} together with the Definition~\ref{def1} of the inverse reduction map from the previous section to present how to check in an easy way whether a given observable is an entanglement witness. Later we show an explicit construction of the new class of entanglement witnesses. We start our consideration from an illustrative example, which shows how to omit checking block-positivity. Let us take Choi-like entanglement witness $W_{Ch}\in \mathcal{B}\left(\mathbb{C}^3 \otimes \mathbb{C}^3 \right) $ from~\cite{Choi}:
\be
\label{choi}
W_{Ch}=\left(\begin{array}{ccc|ccc|ccc}
 \cdot & \cdot & \cdot & \cdot & -1 & \cdot & \cdot & \cdot & -1\\
	\cdot & 1 & \cdot  & \cdot & \cdot & \cdot  & \cdot & \cdot & \cdot\\
	\cdot & \cdot & 1  & \cdot & \cdot & \cdot  & \cdot & \cdot & \cdot\\
	\hline
	\cdot & \cdot & \cdot  & 1 & \cdot & \cdot  & \cdot & \cdot & \cdot\\
	-1 & \cdot & \cdot  & \cdot & \cdot & \cdot  & \cdot & \cdot & -1\\
	\cdot & \cdot & \cdot  & \cdot & \cdot & 1  & \cdot & \cdot & \cdot\\
	\hline
	\cdot & \cdot & \cdot  & \cdot & \cdot & \cdot  & 1 & \cdot & \cdot\\
	\cdot & \cdot & \cdot  & \cdot & \cdot & \cdot  & \cdot & 1 & \cdot\\
	-1 & \cdot & \cdot  & \cdot & -1 & \cdot  & \cdot & \cdot & \cdot\\
\end{array}\right),
\ee
where dots  denote zeros. Now, we show that Theorem~\ref{prop6} together with the property of the reduction map $R$ and its inverse $R^{-1}$ from the previous section implies immediately that $W_{Ch}$ satisfies the second point from the Definition~\ref{def01}, i.e. we show its block-positivity non directly. Namely, we have:
\be
\label{choi2}
\widetilde{W}_{Ch}=\left(\mathbf{1} \otimes \R \right) W_{Ch}=\left(\begin{array}{ccc|ccc|ccc}
	1 & \cdot & \cdot & \cdot & 1 & \cdot & \cdot & \cdot & 1\\
	\cdot & \cdot & \cdot  & \cdot & \cdot & \cdot  & \cdot & \cdot & \cdot\\
	\cdot & \cdot & \cdot  & \cdot & \cdot & \cdot  & \cdot & \cdot & \cdot\\
	\hline
	\cdot & \cdot & \cdot  & \cdot & \cdot & \cdot  & \cdot & \cdot & \cdot\\
	1 & \cdot & \cdot  & \cdot & 1 & \cdot  & \cdot & \cdot & 1\\
	\cdot & \cdot & \cdot  & \cdot & \cdot & \cdot  & \cdot & \cdot & \cdot\\
	\hline
	\cdot & \cdot & \cdot  & \cdot & \cdot & \cdot  & \cdot & \cdot & \cdot\\
	\cdot & \cdot & \cdot  & \cdot & \cdot & \cdot  & \cdot & \cdot & \cdot\\
	1 & \cdot & \cdot  & \cdot & 1 & \cdot  & \cdot & \cdot & 1\\
\end{array}\right)\geq 0.
\ee
Furthermore, we have $\left(\mathbf{1} \otimes R \right)\widetilde{W}_{Ch}=W_{Ch} $.

At the end of this section let us demonstrate how this theory works generalizing the Choi witness. We start from a positive semi-definite operator $0\leq\widetilde{W}\in\mathcal{B}\left(\mathbb{C}^{d}\otimes\mathcal{\mathbb{C}}^{d}\right)$
in the standard operator basis $\mathcal{B}\left(\mathbb{C}^d\right)\ni e_{ij}=\left|i\left\rangle \right\langle j\right|$ for $i,j=1,\ldots,d$:

\begin{equation}
\label{rij}
\widetilde{W}=\sum_{i,j=1}^{d}e_{ij}\otimes\widetilde{W}_{ij}.
\end{equation}
Let $S\in\mathcal{B}\left(\mathbb{C}^{d}\right)$ be a shift operator
defined as:

\[
S\left|i\right\rangle :=\left|i+1\right\rangle \qquad\text{mod}\, d,
\]
then using the above definition we can write operators $\widetilde{W}_{ij}$ from formula~\eqref{rij} in the following way:

\begin{equation}
\label{diag}
\widetilde{W}_{ii}=S^{i-1}\left(\begin{array}{cccc}
a_{1} & 0 & \cdots & 0\\
0 & a_{2} & \cdots & 0\\
\vdots & \vdots & \ddots & \vdots\\
0 & 0 & \cdots & a_{d}
\end{array}\right)S^{(i-1)\dagger},\qquad a_{i}\geq0\quad\text{for} \ : i=1,\ldots,d,
\end{equation}
and for all off-diagonal elements i.e for all indices satisfying $i\neq j$

\begin{equation}
\label{ofdiag}
\widetilde{W}_{ij}=S^{i-1}\left(\begin{array}{cccc}
x & 0 & \cdots & 0\\
0 & 0 & \cdots & 0\\
\vdots & \vdots & \ddots & \vdots\\
0 & 0 & \cdots & 0
\end{array}\right)S^{(j-1)\dagger},\qquad  i,j=1,\ldots,d.
\end{equation}
Using the form of our operator $\widetilde{W}$ from equation~\eqref{rij} together with conditions on $\widetilde{W}_{ij}$ given in formulas~\eqref{diag} and~\eqref{ofdiag} we are able to write explicit conditions for positivity of state $\widetilde{W}$ in terms of parameters $a_i$ and $x$. Namely, we have the following:

\begin{remark}
\label{rem2}
Operator $\widetilde{W}\geq0$ if and only if submatrix $A$ is positive semidefinite

\begin{equation}
A=\left(\begin{array}{cccc}
a_{1} & x & \cdots & x\\
x & a_{1} & \cdots & x\\
\vdots & \vdots & \ddots & \vdots\\
x & x & \cdots & a_{1}
\end{array}\right)\geq0,
\end{equation}
it means that $x\in\left[\frac{-a_{1}}{d-1},a_{1}\right]$.
\end{remark}

Now, we are in the position to use all what we have learnt  from the previous Sections and use the reduction map to construct an appropriate example of entanglement witness. Namely, let us use as a map the inverse of the reduction map i.e $R:\mathcal{B}\left(\mathbb{C}^{d}\right)\rightarrow\mathcal{B}\left(\mathbb{C}^{d}\right)$
defined as follows:

\begin{equation}
R \left(O\right)=\text{Tr}\left(O\right)\mathbf{1}-O.
\end{equation}

The above map is not positive in general. As operator $W$
from Theorem~\ref{prop6} let us take $W=\left(\mathbf{1} \otimes R \right) \widetilde{W}$, then the following conditions should be satisfied

\begin{equation}
\label{eq:condition}
\begin{cases}
x\in\left[\frac{-a_{1}}{d-1},a_{1}\right]\\
x\in\left[-y_{1},\frac{y_{1}}{d-1}\right]\\
y_{k}\geq0 & \text{for}\quad k=1,\ldots,d,
\end{cases}
\end{equation}

where $y_{k}=\frac{1}{d-1}\sum_{i=1}^{d}a_{i}-a_{k},$ for $k=1,\ldots,d$. One can easily see that for dimension $d=3$ and parameters $a_1=0, \ a_2=a_3=1$ we recover Choi-like entanglement witness given by formula~\ref{choi}.

%

\section{Conclusions}
\label{con}
In this paper we have shown that to construct entanglement witnesses it is enough to consider maps which are not  necessarily positive on the whole domain, but only on some sub-domain. We can consider in general non-positive maps (see Theorem~\ref{prop6}) which are surjective functions from the set $\mathcal{P}_{k}^d$  of $k$ rank projectors to the set $\mathcal{P}_1^d$ of rank one projectors (see Corollary~\ref{cor1}, Proposition~\ref{prop3} and Remark~\ref{RemP}). Our illustrative example of such a map is the inverse reduction map (Definition~\ref{def1}) for which we have presented the explicit construction the new class of entanglement witnesses which can be treated as a generalization of Choi entanglement witness.

It is  worth to mention here one open problem connected with our construction. Firstly, it would be interesting to find more surjective maps between sets $\mathcal{P}_k^d$ and $\mathcal{P}_1^d$ in the context of checking whether a given observable $W$ is an entanglement witness without checking its block-positivity. Implementation of our method is much easier in application than checking the above mentioned block-positivity, and for sure for some class of operators (as we illustrated) we can use the statement from Theorem~\ref{prop6} directly which gives an answer almost immediately.
	Secondly,  we can ask about the connection between decomposability property and the structure of the chosen map or chosen operator $W$ (see Theorem~\ref{prop6}).

\section*{Acknowledgments}
 A. Rutkowski was supported by a postdoc internship decision number DEC\textendash{}2012/04/S/ST2/00002, from the Polish National Science Center. M. Studzi{\'n}ski is supported by grant 2012/07/N/ST2/02873 from National Science Centre. M. Mozrzymas is supported by ERC Project no. 291348 QOLAPS and National Science Centre project MaestroDEC-2011/02/A/ST2/00305. M. Mozrzymas would like to thank  Quantum Information Centre of Gda{\'n}sk, where some part of this work was done. Authors would like to thank Agnieszka Rutkowska for valuable comments regarding the first version of the manuscript.

\appendix
\section{Some important facts about unitary spaces}
\label{AppA}

 In this appendix we recall same basic properties of unitary spaces  which are important to understand our results contained in the Section~\ref{results} and the properties of the inverse reduction map from the Section~\ref{reduction}.

\begin{proposition}
\label{prop1}
Let $P$ be an orthogonal projector i.e.
\be
P\in \mathcal{B}\left(\mathbb{C}^d\right):P^{2}=P,\quad P^{\dagger}=P,\quad \tr(P)=d-1,
\ee
then $P$ gives a unique decomposition of the space $\mathbb{C}^{d}$ of the form
\be
\mathbb{C}^{d}=\im P\oplus \ker P:\ker P=(\im P)^{\perp }
\ee
and $\dim (\ker P)=1,$ $\dim (\im P)=d-1$ so $\im P$ is a
hyperplane. Moreover, if $\{|\psi\>_{i}\}_{i=1}^{d-1},\{|\phi\>_{i}\}_{i=1}^{d-1}\subset
\im P$ are two orthonormal bases in the subspace $\im P$ then
\be
P=\sum_{i=1}^{d-1}|\psi_i\>\<\psi_{i}|=\sum_{i=1}^{d-1}|\phi_i\>\<\phi_i|.
\ee
So the spectral decomposition of the projector $P$ does not depend on the
choice of the orthonormal basis in $\im P$.
\end{proposition}

It is known that any set of orthonormal vectors in $\mathbb{C}^{d}$ (or in any linear space) may be extended to a basis of the space $\mathbb{C}^{d}$. Such  extensions are  not unique. The structure of the extensions of
orthonormal bases of the space $\im P$ to bases of the space $\mathbb{C}^{d}$ describes the following:

\begin{lemma}
\label{lem1}
Let $\{|\psi\>_{i}\}_{i=1}^{d-1},\{|\psi_{i}^{\prime }\>\}_{i=1}^{d-1}\subset \im P$
are two orthonormal bases in the subspace $\im P$ and $
\{|\psi\>_{i}\}_{i=1}^{d},\{|\psi_{i}^{\prime }\>\}_{i=1}^{d}\subset \mathbb{C}^{d}$ are their extensions to orthonormal bases in $\mathbb{C}^{d}$, then
\be
|\psi_{d}^{\prime }\>=e^{i\varphi }|\psi_{d}\>,\quad \varphi \in \lbrack 0,2\pi ),
\ee
where $|\psi_{d}^{\prime }\>,|\psi_{d}\>\in \ker P.$ So it means that for a given
orthonormal projector $P$ of rank $d-1$ there exists a vector $|\psi\>\in \ker P$
such that $||\psi||=1$ and the extension of any orthonormal basis $
\{|\psi_{i}\>\}_{i=1}^{d-1}$ in $\im P$ to a basis in $\mathbb{C}^{d}$ has the following form
\be
\{|\psi_{1}\>,...,|\psi_{d-1}\>,e^{i\varphi }|\psi\>\}.
\ee
Any vector of the form $e^{i\varphi }|\psi\>$ where $|\psi\>\in \ker P$, $||\psi||=1$ form
an orthonormal basis in one-dimensional subspace $\ker P$.
\end{lemma}

\begin{proof}
$|\psi_{d}^{\prime}\>=\sum_{i=1}^{d}x_{i}|\psi_{i}\>$. From $|\psi_{d}^{\prime }\>\perp
|\psi_{i}\>, $ $i=1,...,d-1$ we get $|\psi_{d}^{\prime }\>=x_{d}|\psi_{d}\>$ and from the
normalization of basis vectors we get $|x_{d}|=1.$
\end{proof}

From this Lemma and Proposition 1 we get

\begin{corollary}
\label{cor1}
Let us define
\be
\mathcal{P}_{d-1}^{d}=\{P\in \mathcal{B}\left(\mathbb{C}^d\right):P^{2}=P,\quad P^{\dagger}=P\qquad \tr(P)=d-1\},
\ee
\be
\mathcal{P}_{1}^{d}=\{P\in \mathcal{B}\left(\mathbb{C}^d\right):P^{2}=P,\quad P^{\dagger}=P\qquad \tr(P)=1\}.
\ee
There exists a unique bijective correspondence between the elements of the sets $
\mathcal{P}_{d-1}^{d}$ and $\mathcal{P}_{1}^{d}$. The bijective correspondence between the elements of the sets $
\mathcal{P}_{d-1}^{d}$ and $\mathcal{P}_{1}^{d}$ can be expressed as
follows
\be
\forall P\in \mathcal{P}_{d-1}^{d}\quad \exists !Q\in \mathcal{P}_{1}^{d}\quad P=\mathbf{1}-Q,
\ee
\be
\forall Q\in \mathcal{P}_{1}^{d}\quad \exists !P\in \mathcal{P}_{d-1}^{d}\quad Q=\mathbf{1}-P.
\ee

Moreover, if $P\in \mathcal{B}\left(\mathbb{C}^d\right):P^{2}=P,\quad P^{\dagger}=P,\qquad \tr(P)=d-1$ then there exists a unique
orthonormal projector $Q\in \mathcal{B}\left(\mathbb{C}^d\right):Q^{2}=Q,\quad Q^{\dagger}=Q,\qquad \tr(Q)=1$ such that
\be
\mathbf{1}=P+Q:Q=|\psi\>\<\psi|,
\ee
where $|\psi\>\in \ker P$ is any orthonormal basis vector of $\ker P$ and $\im Q=\ker P,$ $\ker Q=\im P$ and $PQ=QP=0$.

\begin{proof}
To prove the second statement let us consider orthonormal basis $\{|\psi_{i}\>\}_{i=1}^{d}$ in $\mathbb{C}^{d}$ we have
\be
\mathbf{1}=\sum_{i=1}^{d}|\psi_i\>\<\psi_i|.
\ee
In particular it holds for orthonormal bases of $\mathbb{C}^{d}$ that are extensions of the orthonormal bases of $\im P$ e.i. for
the bases of the form $\{|\psi_{1}\>,...,|\psi_{d-1}\>, |\psi\>\}$, where $\{|\psi_{1}\>,...,|\psi_{d-1}\>\}
$ is an orthonormal basis in $\im P$ \ and $|\psi\>\in \ker P$ : $||\psi||=1$
forms an orthonormal basis in one-dimensional  $\ker P$ so we have
\be
\mathbf{1}=\sum_{i=1}^{d-1}|\psi_i\>\<\psi_i|+|\psi\>\<\psi|\equiv P+Q,
\ee
where $Q=|\psi\>\<\psi|$ and from Proposition 1 we know that the orthogonal
projectors does not depend on the choice of bases in the range of these
projectors so $Q$ do not depend on the choice of the basis vector $|\psi\>$ and
is unique.
\end{proof}
\end{corollary}

\section{Auxiliary lemmas}
\label{AppB}
After  a short introduction to the topic of unitary spaces contained in the Section~\ref{AppA} we are ready to present a conclusion which is contained in the two following propositions.
First, the Proposition~\ref{prop3} contains  a generalization of the bijection from the Corollary~\ref{cor1} for the  rank $k$ projectors, which allows us to formulate a general statement contained in the Theorem~\ref{prop6}. Finally, the Proposition~\ref{prop5} is an auxiliary result important in the proof of the above-mentioned theorem.




\begin{proposition}
\label{prop3}
Let $P$ be an orthogonal projector i.e.
\be
P\in \mathcal{B}\left(\mathbb{C}^d\right):P^{2}=P,\quad P^{\dagger}=P\qquad \tr(P)=k,\quad k=1,..,d-1,
\ee
then $P$ gives a unique decomposition of the space $\mathbb{C}^{d}$ of the form
\be
\mathbb{C}^{d}=\im P\oplus \ker P:\ker P=(\im P)^{\perp }
\ee
and $\dim (\ker P)=d-k,$ $\dim (\im P)=k.$ Moreover, for any such $P$
there exists a unique orthogonal projector
\be
Q\in \mathcal{B}\left(\mathbb{C}^d\right):Q^{2}=Q,\quad Q^{\dagger}=Q,\qquad \tr(Q)=d-k,
\ee
such that
\be
\mathbf{1}=P+Q,
\ee
where $\im Q=\ker P,$ $\ker Q=\im P$ and $PQ=QP=0$, so we have
\be
\forall k=1,..,d-1~\forall P\in \mathcal{P}_{k}^{d}\quad \exists !Q\in
\mathcal{P}_{d-k}^{d}\quad P=\mathbf{1}-Q,
\ee
where  $\mathcal{P}_{k}^{d}=\{P\in
\mathcal{B}\left(\mathbb{C}^d\right):P^{2}=P,\quad P^{\dagger}=P,\quad \tr(P)=k\}$.
\end{proposition}

\begin{remark}
\label{RemP}
Reader notices that for our purposes in the Theorem~\ref{prop6} we can choose  bijection which establishes one to one correspondence between set $\mathcal{P}_k^d$ of rank $k$ projectors and the set $\mathcal{P}_1^d$ of rank one projectors.
\end{remark}

In the following we will need also

\begin{proposition}
\label{prop5}
\bigskip Let $|\psi\>\in\mathbb{C}^{d}$ and $|\phi_{i}\>\in\mathbb{C}^{d}$, $i=1,..,d-1$ are such that
\be
\<\psi|\psi\>=1,~\<\phi_{i}|\phi_{j}\>=\delta _{ij}.
\ee
Then
\be
P=|\psi\>\<\psi|\otimes
\sum_{i=1}^{d-1}|\phi_{i}\>\<\phi_{i}|=\sum_{i=1}^{d-1}|\omega_{i}\>\<\omega_{i}|\in \mathcal{B}\left(\mathbb{C}^d\ot \mathbb{C}^d\right),
\ee
where $\omega_{i}=|\psi\>\otimes |\phi_{i}\>,$ \ is an orthogonal projector of rank $d-1$ (in
fact $P\in \mathcal{P}_{d-1}^{d^{2}})$ and it is generated by simple tensors
$\omega_{i}$, so it is of a particular form. Note that
\be
\left\{\sum_{i=1}^{d-1}|\phi_{i}\>\<\phi_{i}|:\<\phi_{i}|\phi_{j}\>=\delta _{ij}\right\}=\mathcal{P}
_{d-1}^{d}\in \mathcal{B}\left(\mathbb{C}^d\right).
\ee
\end{proposition}


\begin{thebibliography}{9}

\bibitem{Bennett} C.~H. Bennett, G. Brassard \textit{Quantum Cryptography: Public key distribution and coin tossing}, Proc. Internat. Conf. Computer Systems and Signal Processing, Bangalore, {\bf 175}, (1984).

\bibitem{Bennett2} C.~H. Bennett, G. Brassard. Cr\'epeau, R. Jozsa, A. Peres, W.~K. Wootters,  \textit{Teleporting an Unknown Quantum State via Dual Classical and Einstein-Podolsky-Rosen Channels}, Phys. Rev. Lett. {\bf 70} 1895-1899 (1993).

\bibitem{Bennett3} C.~H. Bennett, S. Wiesner \textit{Communication via one- and two-particle operators on Einstein-Podolsky-Rosen states}, Phys. Rev. Lett. {\bf 69} (20), 2881, (1992).

\bibitem{Cho92} S. J. Cho, S-H. Kye i S. G. Lee. \textit{Generalized choi maps in three-dimensional
matrix algebra}, Linear Algebra Appl. {\bf 171}, 213, (1992).

\bibitem{Hor1} M. Horodecki, P. Horodecki and R. Horodecki, \textit{%
Separability of mixed states: Necessary and sufficient conditions}, Phys.
Lett. A, vol. {\bf 223}, pp. 1-8, 1996





\bibitem{Peres} A. Peres, \textit{Separability criterion for density matrices%
}, Phys. Rev. Lett., vol. {\bf 77}, pp. 1413-1415, 1996

\bibitem{Terhal} B.~M. Terhal, \textit{Bell inequalities and separability criterion}, Phys. Lett. A {\bf 271}, 319 (2000)

\bibitem{PH} P. Horodecki, \textit{Separability criterion and inseparable mixed states
with positive partial transposition}, Phys. Rev. Lett. A {\bf 232}, 333,
(1997).

\bibitem{Doherty1} A.~C. Doherty, P.~A. Parrilo, and F.~M. Spedalieri \textit{Distinguishing separable and entangled states
}, Phys. Rev. Lett. {\bf 88}, 187904 (2002).

\bibitem{Doherty2} A.~C. Doherty, P.~A. Parrilo, and F.~M. Spedalieri, \textit{A complete family of separability criteria}, Phys. Rev. A {\bf 69}, 022308 (2004)

\bibitem{Guhne} O. G{\"u}hne, G. T\'oth, \textit{Entanglement detection}, Phys. Rep.  {\bf 474}, 1 (2009)

\bibitem{Hor2} R. Horodecki, P. Horodecki, M. Horodecki and K. Horodecki, \textit{%
Quantum entanglement}, Rev. Mod. Phys. {\bf 81}, 865, (2009)

\bibitem{Hor3} M. Horodecki, P. Horodecki and R. Horodecki, \textit{Reduction criterion of separability and limits for a class of distillation protocols}, Phys. Rev. A {\bf 59}, 4206, (1999).

\bibitem{Sper} J. Sperling and W. Vogel, \textit{Necessary and sufficient conditions for bipartite
entanglement}, Phys. Rev. {\bf 79}, 022318, (2009).

\bibitem{Toth} G. T\'oth, \textit{Entanglement witnesses in spin models}, Phys. Rev. A {\bf 71},
010301(R), (2005).

\bibitem{Jam} A. Jamio{\l}kowski, \textit{Linear transformations which preserve trace and positive semi-definiteness of operator}, Rep. Math. Phys. {\bf 3}, 267-278, (1972).

\bibitem{chrust} D. Chru{\'s}ci{\'n}ski, G. Sarbicki, \textit{Entanglement witnesses: construction, analysis and classification}, J. Phys. A: Math. Theor.  {\bf 47}, 483001, (2014).

\bibitem{Horn} R. A. Horn, Ch. R. Johnson \textit{Topics in Matrix Analysis}, Cambridge University Press, 1991

\bibitem{Choi} M.-D. Choi, \textit{Positive linear maps on c algebras}, Cand. J. Math. {\bf 24}, 520, (1972).










\end{thebibliography}
\end{document}